\documentclass[12pt]{article}
\usepackage[margin=1in]{geometry}
\usepackage{graphicx,amssymb,amsmath}
\usepackage{color}
\usepackage{thmtools, hyperref}
\usepackage{algorithm2e}

\newtheorem{theorem}{Theorem}[section]
\newtheorem{lemma}[theorem]{Lemma}
\newtheorem{proposition}{Proposition}
\newtheorem{corollary}[theorem]{Corollary}

\newenvironment{proof}[1][Proof]{\begin{trivlist}
\item[\hskip \labelsep {\bfseries #1}]}{\end{trivlist}}

\newcommand{\qed}{\nobreak \ifvmode \relax \else
      \ifdim\lastskip<1.5em \hskip-\lastskip
      \hskip1.5em plus0em minus0.5em \fi \nobreak
      \vrule height0.75em width0.5em depth0.25em\fi}

\begin{document}
\title{Unfolding Genus-2 Orthogonal Polyhedra\\ 
with Linear Refinement}
\author{Mirela Damian%
\thanks{Department of Computer Science, Villanova University, Villanova, PA 19085, USA. \protect\url{mirela.damian@villanova.edu}}
\and
Erik Demaine%
\thanks{Computer Science and Artificial Intelligence Laboratory, 
Massachusetts Institute of Technology, 
32 Vassar St., Cambridge, MA 02139, USA. 
\protect\url{edemaine@mit.edu}}
\and
Robin Flatland
\thanks{Department of Computer Science, 
Siena College, 
Loudonville, NY 12211, USA.
\protect\url{flatland@siena.edu}}
\and
Joseph O'Rourke%
\thanks{Department of Computer Science, 
Smith College, Northampton, MA 01063, USA. 
\protect\url{jorourke@smith.edu}}
}

\date{}
\maketitle

\begin{abstract}
We show that every orthogonal polyhedron of genus $g \le 2$ can be unfolded
without overlap while using only a linear number of orthogonal cuts (parallel to the polyhedron edges). 
This is the first result on unfolding general orthogonal polyhedra beyond genus-0. Our 
unfolding algorithm  
relies on the existence of at most $2$ special leaves in what we call the ``unfolding tree'' 
(which ties back to the genus), so unfolding polyhedra of genus $3$ and beyond requires new techniques. 
\end{abstract}


\section{Introduction}
\label{sec:intro}
An \emph{unfolding} of a polyhedron is produced by cutting its surface in such a way that it can be flattened to 
a single, connected piece without overlap.
In an \emph{edge unfolding}, the cuts are restricted to the polyhedron's edges,
whereas in a \emph{general unfolding}, cuts can be made anywhere on the surface.
It is known that edge cuts alone are not sufficient to
guarantee an unfolding for non-convex polyhedra~\cite{Bern-Demaine-Eppstein-Kuo-Mantler-Snoeyink-2003,BDDLOORW1998},
and yet it is an open question as to whether all non-convex
polyhedra have a general unfolding.
In contrast, it is unknown 
whether every convex polyhedron has an edge unfolding~\cite[Ch.~22]{Demaine-O'Rourke-2007}, but all convex polyhedra
have general unfoldings~\cite[Sec.~24.1.1]{Demaine-O'Rourke-2007}.

The successes to date in unfolding non-convex objects have been
with the class of orthogonal polyhedra. This class consists of polyhedra whose edges and faces all meet at right angles.
Because not all orthogonal polyhedra have
edge unfoldings (even for simple examples such as a box with a smaller
box extruding out on top)~\cite{BDDLOORW1998}, the unfolding algorithms use additional non-edge cuts.
These additional cuts generally follow one of two models.
In the \emph{grid unfolding model}, the orthogonal polyhedron is sliced by axis perpendicular planes passing through each vertex, and cuts are allowed along
the slicing lines where the planes intersect the surface.
In the \emph{grid refinement model}, each rectangular grid face
under the grid unfolding model is further subdivided by
an $(a \times b)$ orthogonal grid, for some positive integers $a, b \ge 1$, and cuts are also allowed
along any of these grid lines.

There have been three phases of research on unfolding orthogonal polyhedra.
The first phase focused on unfolding special subclasses, which included
orthotubes~\cite{BDDLOORW1998}, well-separated orthotrees~\cite{Damian-Flatland-Meijer-O'Rourke-2005-orthotrees}, orthostacks~\cite{BDDLOORW1998,Damian-Meijer-2004-orthostacks}, and Manhatten towers~\cite{Damian-Flatland-O'Rourke-2008-manhattan}.
These algorithms use the grid unfolding
model or the grid refinement model with a constant amount of
refinement (i.e., $a$ and $b$ are both constants).

The second phase began with the discovery of the epsilon-unfolding algorithm~\cite{Damian-Flatland-O'Rourke-2007-epsilon} which unfolds all genus-0 orthogonal polyhedra. A key component of the unfolding algorithm is the determination of a spiral path on the surface of the polyhedron that unfolds to a planar monotone staircase, from which the rest of the surface attaches (without overlap) above and below. A drawback of that algorithm is that it requires an exponential amount of grid refinement. Subsequent improvements, however, reduced the amount of refinement to quadratic~\cite{Damian-Demaine-Flatland-2014-delta}, and then to linear~\cite{Chang2015}, with both algorithms following the basic outline of~\cite{Damian-Flatland-O'Rourke-2007-epsilon}.

The third phase of research addresses the next obvious challenge, that of unfolding
higher genus polyhedra. To our knowledge, the
only attempt at this is that of Liou et al.~\cite{Liou-Poon-Wei-2014-onelayer}.
They provide an algorithm for unfolding a special subclass of one-layer orthogonal polyhedra in which all faces are unit squares and the holes are unit
cubes.

Thus the question of whether all orthogonal polyhedra of genus greater than zero can be unfolded is still wide open, and is in a sense the natural endpoint of this line of investigation.
In this paper we take a significant step toward this goal by presenting a new algorithm that unfolds all orthogonal polyhedra of genus $1$ or $2$. The algorithm extends ideas from~\cite{Chang2015} by making several key modifications to circumvent issues that arise from the presence of holes. As in~\cite{Chang2015}, our algorithm only
requires linear refinement.
         
\subsection{Notation and Definitions}
Let $P$ be an orthogonal polyhedron of genus $g \le 2$, whose edges are parallel to the coordinate axes and whose 
surface is a $2$-manifold. 
We take the $z$-axis to define the \emph{vertical} direction, 
the $x$-axis to determine \emph{left} and \emph{right},
and the $y$-axis to determine \emph{front} and \emph{back}.
We consistently take the viewpoint from $y=-\infty$.
The faces of $P$ are distinguished by their outward normal:
forward is $-y$; rearward is $+y$; left is $-x$; right is $+x$; bottom is $-z$; and top is $+z$.%
\footnote{The $\pm y$ faces are given the awkward names
  ``forward'' and ``rearward'' to avoid confusion with other uses
  of ``front'' and ``back'' introduced later.}

\begin{figure}[htbp]
\centering
\includegraphics[width=.98\linewidth]{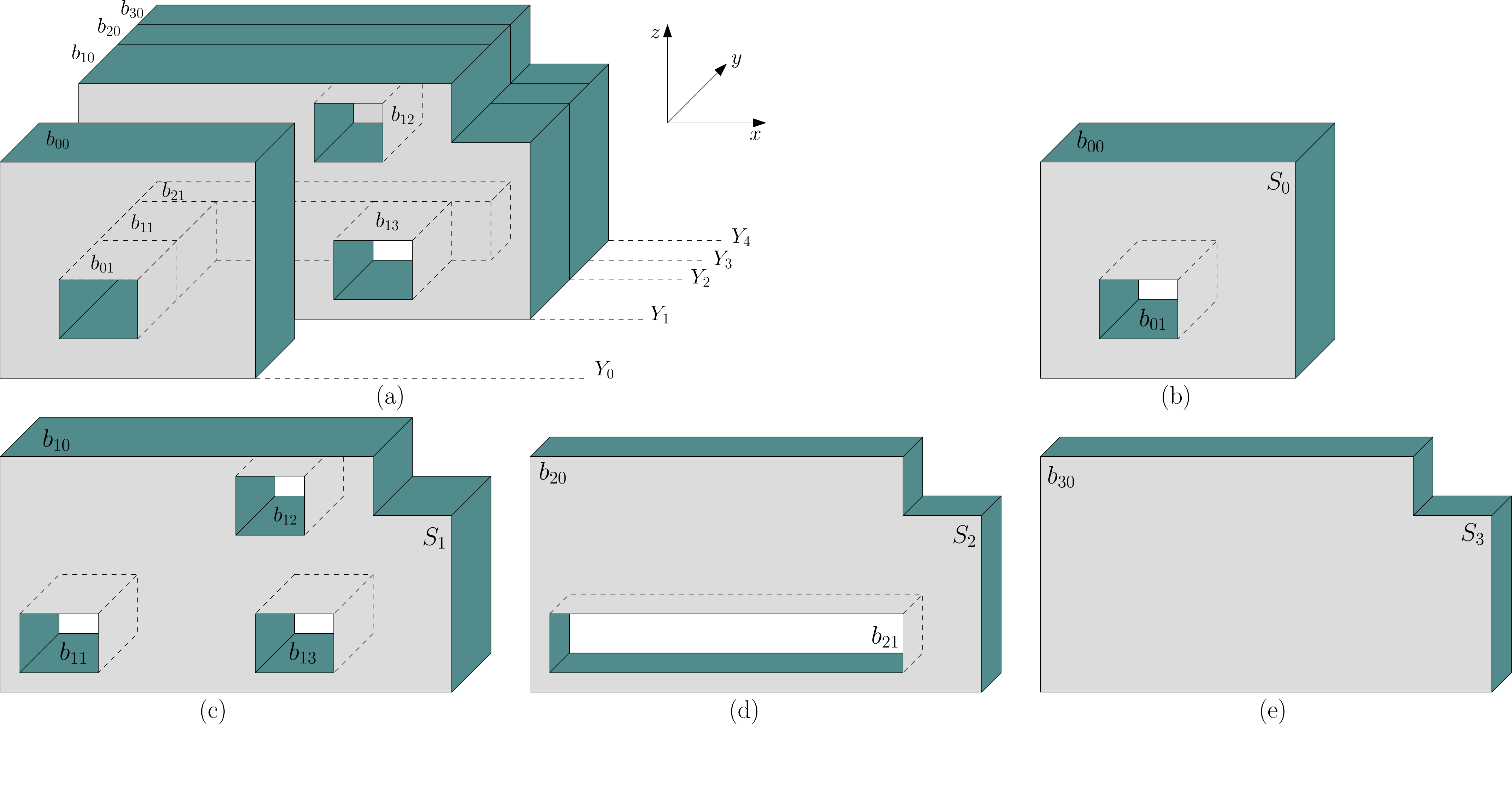}
\caption{A polyhedron of genus one.}
\label{fig:defs1}
\end{figure}

Imagine slicing $P$ with $y$-perpendicular planes through each vertex. Let $Y_0$, $Y_1$, $Y_2$, $\dots$
be the slicing planes sorted by $y$ coordinate. Each (solid) connected component of $P$ located 
between two consecutive planes $Y_i$ and $Y_{i+1}$ is called a \emph{slab}. 
For example, the polyhedron from~\autoref{fig:defs1}a has four slabs $S_0$, $S_1$, $S_2$ and $S_3$, 
which are depicted in~\autoref{fig:defs1}(b-e). 
Note that each slab is an extruded orthogonal polygon with zero or more orthogonal holes, extruded in the $y$-direction. 
The cycle of \{left, right, top, bottom\} faces surrounding either the entire slab or a hole in a slab is called a \emph{band}. 
Each slab has exactly one \emph{outer} band, and zero or more \emph{inner} bands. Referring to the example from~\autoref{fig:defs1}, 
the slab $S_0$ has outer band $b_{00}$ and inner band $b_{01}$; $S_1$ has outer band $b_{10}$ and inner bands $b_{11}$, $b_{12}$ and $b_{13}$; 
$S_2$ has outer band $b_{20}$ and inner band $b_{21}$; and $S_3$ has outer band $b_{30}$ and no inner bands.
Note that each band is associated with a unique slab. 
The intersection of a band with an adjacent plane $Y_i$ (and similarly in $Y_{i+1}$) is a cycle of edges called a \emph{rim} 
(so each band has exactly two rims). 

We say that a rim $r$ \emph{encloses} a face of $P$ if
the portion of the $Y$-plane interior to $r$  
is a face of $P$.  In other words, all points enclosed by $r$ 
in the $Y$-plane are also on the surface of $P$. 
If there are points of the $Y$-plane enclosed by $r$ that are not on the surface of $P$, then 
we say that $r$  does not enclose a face of $P$.  For example,
in~\autoref{fig:defs1} the rim of $b_{30}$ in plane $Y_4$ and the rim of $b_{12}$ in plane $Y_2$ each
enclose a
face of $P$, but the rim of $b_{00}$ in $Y_0$ does not. 

\section{Overview of Linear Unfolding}
\label{sec:genus0}
Throughout this section, $P$ is an orthogonal polyhedron of genus zero. 
We begin with an overview of the algorithm in~\cite{Chang2015} that unfolds $P$ using linear refinement. 
In~\autoref{sec:genus2} we will detail those aspects of the algorithm that we modify to handle orthogonal 
polyhedra of genus 1 and 2.

\subsection{Unfolding Extrusions}
\label{sec:extrusions}
Nearly all algorithmic issues in the linear unfolding algorithm from~\cite{Chang2015} are present in unfolding polyhedra that 
are vertical extrusions of simple orthogonal polygons.
Therefore, we describe their unfolding algorithm for this simple shape class first, 
before extending the ideas to all orthogonal polyhedra of genus zero.  

Before going into details, we briefly describe the algorithm
at a high level.  It begins by slicing $P$ into slabs using $y$-perpendicular
planes.  
For these vertical extrusions, all the slabs are boxes.
The adjacency graph of these boxes is a tree $T$. 
Each leaf node $b$ in $T$ has a
corresponding thin spiral surface path that includes a vertical segment running across the back face of $b$ (which must be a face of $P$) on the side opposite 
to $b$'s parent. The surface path extends from the bottom endpoint of this vertical segment by cycling around $b$'s band until it 
reaches the top endpoint, and from there it continues
along two strands that spiral side-by-side together on $P$, cycling 
around the bands
on the path in $T$ to the root node box where the two strands terminate.  
At the root box, the endpoints of all the pairs of strands are carefully stitched together into one surface path
 that  
 can be flattened in the plane as 
 a monotone staircase.  By thickening the surface path to cover
the band faces of $P$,
and attaching the $y$-perpendicular faces 
above and below it, the entire
surface of $P$ is flattened without overlap into the plane. Details
of the algorithm are provided in the sections that follow. 

\subsubsection{Unfolding Tree $T$}
\label{sec:extrusions-unftree}
Again we restrict our attention to the situation where $P$ is 
a vertical extrusion of a simple orthogonal polygon, and describe the algorithm in detail.
The unfolding algorithm begins by slicing $P$ with $Y_i$ planes passing through every vertex, as
described in~\autoref{sec:intro}. 
This induces a partition of $P$ into rectangular boxes. 
See~\autoref{fig:g0poly}a for an example.
The dual graph of this partition is a tree $T$ whose nodes correspond to 
bands, and whose edges connect pairs of adjacent bands. 
\autoref{fig:g0poly}b shows the tree $T$ for the example from~\autoref{fig:g0poly}a. 
We refer to $T$ as the {\em unfolding tree}, since it will guide the unfolding process. 
For simplicity, we will use the terms ``node'' and ``band'' interchangeably. 

\begin{figure}[htbp]
\centering
\includegraphics[width=0.95\linewidth]{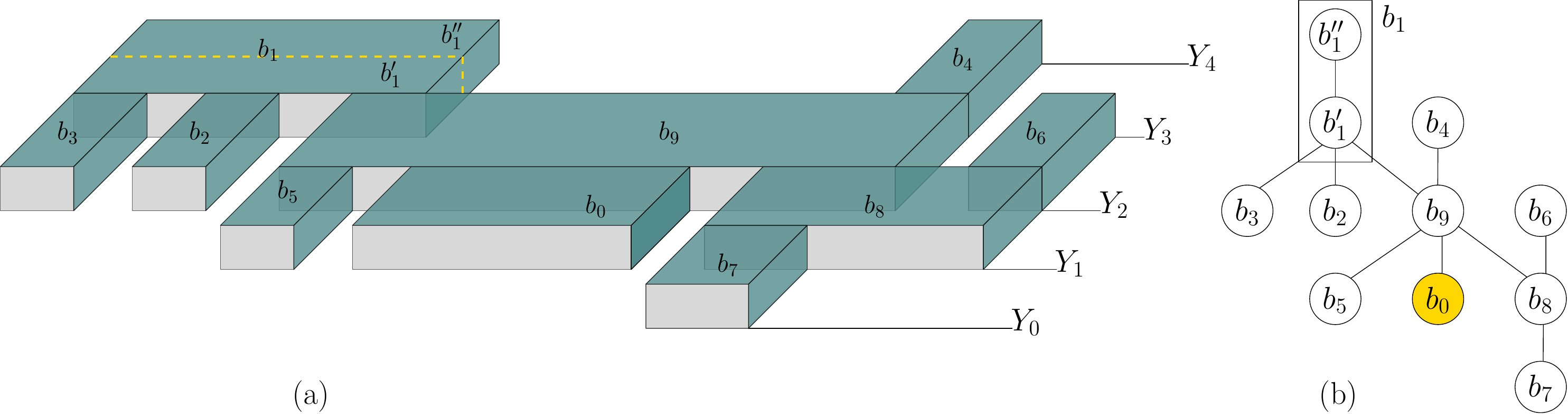}
\caption{(a) Genus-0 polyhedron partitioned into boxes (b) Unfolding tree $T$ rooted at $b_0$.} 
\label{fig:g0poly}
\end{figure}

Every tree of two or more nodes has at least two nodes of degree one,
so we designate the root of $T$ to be one of these degree-$1$ nodes. 
In~\autoref{fig:g0poly}b for example, we may choose $b_0$ as the root of 
$T$, although any of the other degree-$1$ nodes 
would serve as well. 
The rim of the root band that has no adjacent band 
is called its \emph{front} rim; the other rim 
is its {\em back} rim.  
For any other band $b$, the rim adjacent to $b$'s parent in $T$ is the {\em front} rim of $b$, 
and the other rim of $b$ is its {\em back} rim. Children attached to the front rim of their parent are 
\emph{front children}; children attached along the back rim of their parent are \emph{back children}.
Note that ``front'' and ``back'' modifiers for rims and children derive from the structure of
$T$, and are not related to the ``forward'' and ``rearward'' $\pm y$
directions. 

For reasons that will become clear later, we slightly alter the structure of $T$ to eliminate all non-leaf
nodes without back children. For each such internal band $b$, we perform a cut around its middle with a 
$y$-perpendicular plane. This partitions $b$ into two bands $b'$ and $b''$, with $b'$ at the front 
and $b''$ at the back of $b$. This change in the partition is mirrored in $T$ by replacing $b$ with $b'$, 
and adding $b''$ as a back leaf child of $b'$. In~\autoref{fig:g0poly}, node $b_1$ is replaced by $b'_1$ and $b''_1$.
Thus each non-leaf node in $T$ has at least one back child. 

\subsubsection{Leaf Node Unfolding}
\label{sec:leaf}


The unfolding of a leaf node $b$ is determined by a spiral surface path 
whose endpoints lie on a top rim segment shared by $b$
and its parent (necessarily on $b$'s front rim, by definition).
See~\autoref{fig:leaf}a where the 
endpoints are labeled $e_1$ and $e_2$.
Observe that the middle of the path consists of a vertical segment on $b$'s back face, shown in red 
on the exploded view of the back face in~\autoref{fig:leaf}a and circled in~\autoref{fig:leaf}b. 
We describe the spiral path as it extends out from the top and bottom
of this segment to connect up with the 
endpoints on the front rim. 
From the bottom of the segment, the path moves parallel to the $y$-axis on the bottom face and then cycles 
counterclockwise to 
the top face where it meets up with the top end of the vertical segment. From there, both ends of the path
cycle side-by-side together in a counterclockwise direction 
while displacing toward the front rim.
We refer to this spiral path as the \emph{connector path}, suggestive of its ability to connect two points ($e_1$ and $e_2$) that are
both located on the same rim. 
When unfolded and laid horizontally in the plane, this spiral forms a monotone staircase, as depicted in~\autoref{fig:leaf}b. 

\begin{figure}[htbp]
\centering
\includegraphics[width=0.8\linewidth]{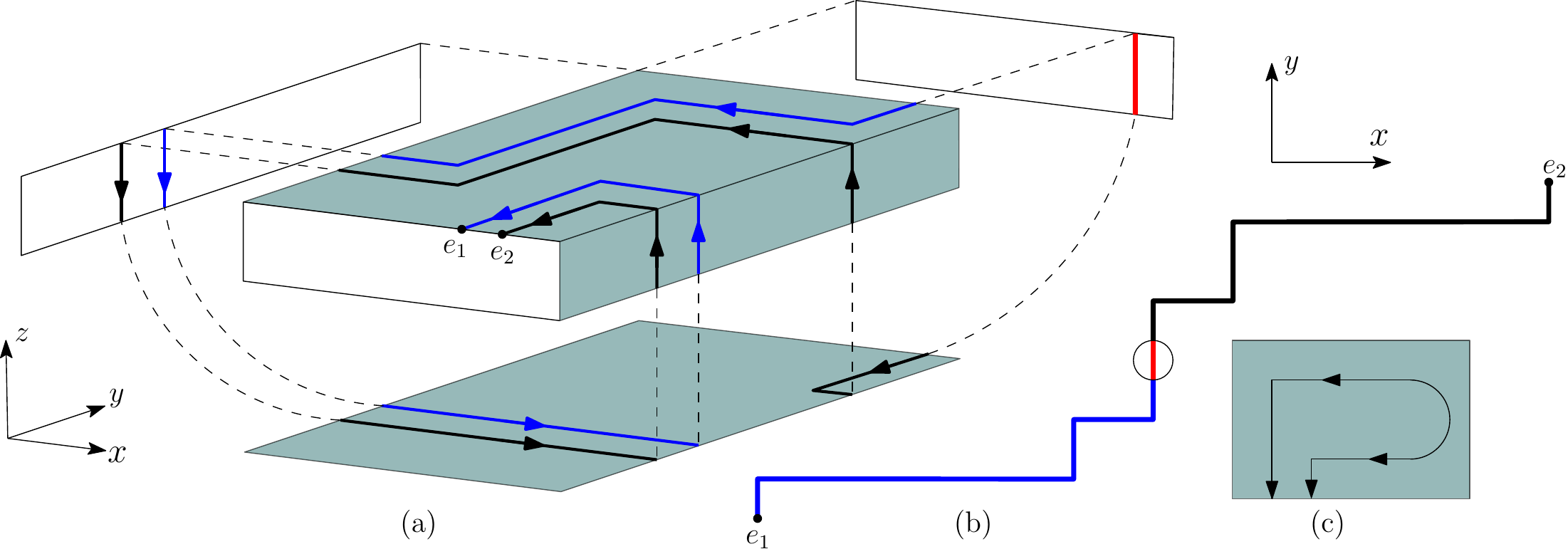}
\caption{Unfolding a leaf band in counterclockwise direction; arrows indicate the direction followed by the unfolding algorithm, starting at the 
back face vertical segment (shown in red in (a) and circled in (b) of this figure).} 
\label{fig:leaf}
\end{figure}

Three-dimensional illustrations, such as the one in~\autoref{fig:leaf}a, become impractical for more complex examples, so we will use instead  
the 2D representation depicted in~\autoref{fig:leaf}c. 
This 2D representation captures the counterclockwise direction of the blue and black portions of the path in~\autoref{fig:leaf}a,
viewed from $y=-\infty$ 
as they cycle side-by-side together
from the back
to the front rim; the arc symbolizes the vertical back face segment connecting them.

A crucial property required by the Chang and Yen's unfolding algorithm~\cite{Chang2015} and 
the algorithms from which that derives~\cite{Damian-Flatland-O'Rourke-2007-epsilon,Damian-Demaine-Flatland-2014-delta} is
that the back rim of each leaf band in $T$ encloses a face of $P$. This is necessary
because the connector paths use a thin strip 
from the back faces of the leaves. Although it is easy to verify this property for the
simple shape class of extrusions, it is not obvious for
arbitrary genus-0 orthogonal polyhedra, but was proven true in~\cite{Damian-Flatland-O'Rourke-2007-epsilon}. 

\subsubsection{Internal Node Unfolding}
\label{sec:internal}
Having established a spiral path for each leaf node, we then
extend these paths to the {\emph{internal}
nodes in $T$, where an
internal node is any non-leaf node other than the root. 
We process internal nodes of $T$ in order of increasing height of their corresponding subtrees. This guarantees that,
at the time an internal node $b$ is processed, its children in $T$ have already been processed.  We assume
inductively that having processed a front (back) child of $b$, the two endpoints of each spiral path originating at a leaf in the child's subtree are located side-by-side on the front (back) rim of $b$. The goal in processing $b$ is to extend these paths so that the pairs of endpoints
lie side-by-side on the top of $b$'s rim segment shared with $b$'s parent. 
In~\autoref{fig:leaf} for example, $b''_1$ would have already been processed at the 
time $b'_1$ is processed, 
and the paths need to be extended across $b'_1$ to the front
rim of $b'_1$ shared with its parent $b_9$.
The total number of spiral paths handled at $b$ is precisely the 
number of leaves in the subtree of $T$ rooted at $b$. 

Let $r$ be the top rim edge shared by $b$ with its parent in $T$. Refer to the band labeled $b$ and the edge labeled $r$ in~\autoref{fig:unf0} (which shows the unfolding for the example from~\autoref{fig:g0poly}). Let $\xi_1, \xi_2, \ldots,$ be the spiral paths corresponding to the back children of $b$, listed in the order in which they are encountered in a clockwise walk starting at the top left corner of $b$'s back rim). 
Our construction of $T$ guarantees that at least one such back spiral exists at each internal node in $T$. 

For each $i = 1, 2, \ldots$, we extend both ends of $\xi_i$ by tracing along both sides of an orthogonal path that makes one complete counterclockwise cycle around the top, left, bottom and right faces of $b$, while displacing toward the front of $b$, until it reaches $r$. The complete cycle around $b$ is important to ensure that the spiral can later be thickened to cover the entire surface of $b$ (hence the need for at least one back child at each internal node). 
For $i >1$, the orthogonal path corresponding to $\xi_i$ runs alongside the orthogonal path corresponding to $\xi_{i-1}$, to ensure that the spiral paths do not cross one another. (In~\autoref{fig:unf0}, the spiral paths are labeled in several places, to 
permit easy tracing. 
Because $\xi_i$ only moves 
parallel to the $y$ axis and spirals counterclockwise around $b$, it can be laid flat in the plane as a staircase monotone in the $x$-direction. 

\begin{figure}[htbp]
\centering
\includegraphics[width=0.8\linewidth]{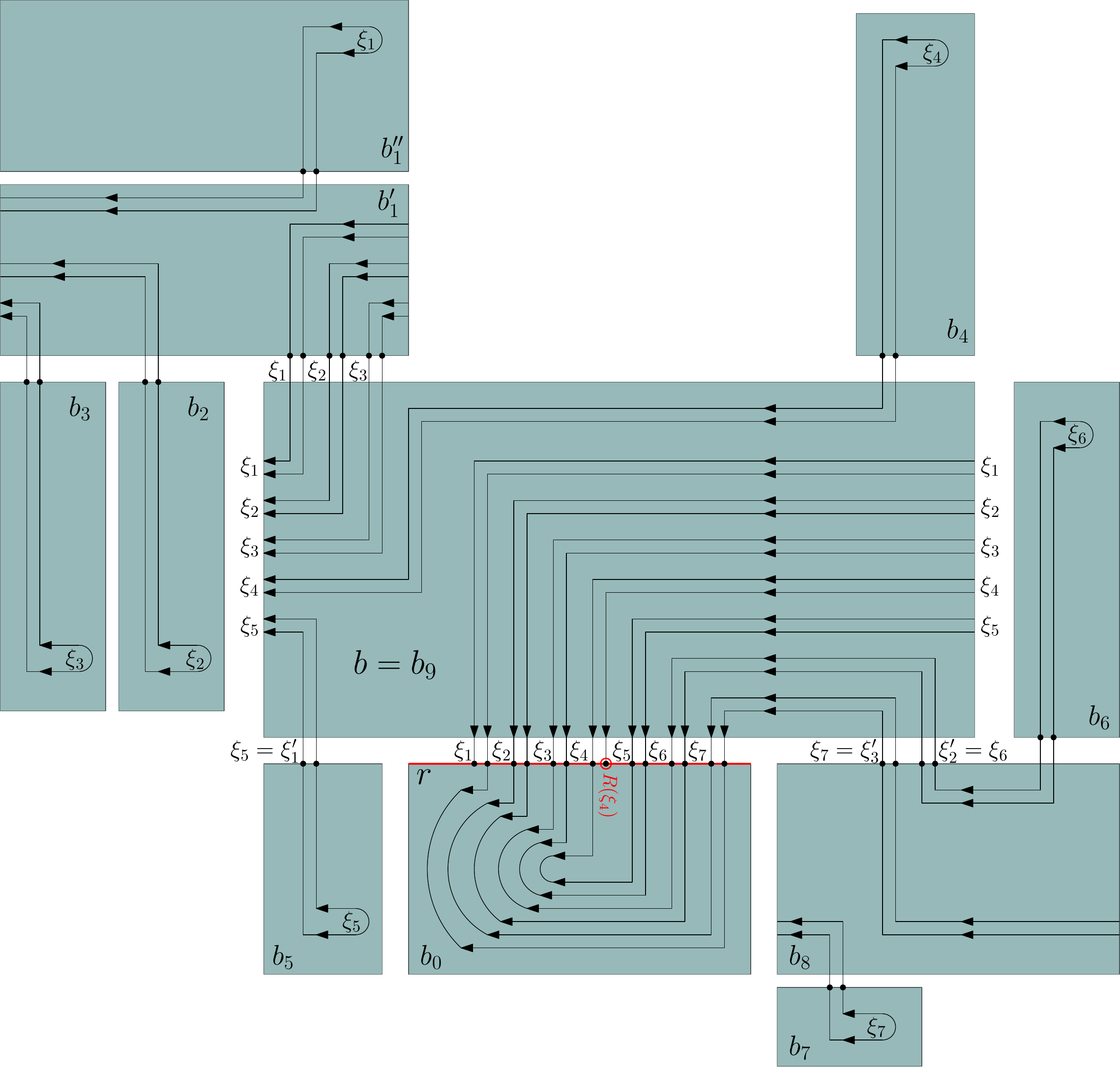}
\caption{Unfolding internal nodes and the root band ($b_0$) for the example from~\autoref{fig:g0poly}; arrows indicate the direction followed by the unfolding algorithm.}
\label{fig:unf0}
\end{figure}

We now turn to processing the spiral paths corresponding to the front children of $b$. 
Let $\xi'_1, \xi'_2, \ldots,$ be the front spiral paths encountered in this order in a counterclockwise walk around the front rim of $b$, starting at any point on the front rim of $b$. 
We extend both ends of each spiral $\xi'_i$ by tracing along both sides of an orthogonal path that displaces slightly toward the back of $b$, then proceeds counterclockwise and toward the front of $b$ until it meets $r$. Note that, if $\xi_i$ lies to the left of $r$, then it will need to cycle around the top, left, bottom and right faces of $b$, in order to meet $r$ (see spiral $\xi'_1 = \xi_5$ in~\autoref{fig:unf0}). Again, care must be taken to ensure that the orthogonal path corresponding to $\xi'_i$ does not cross any of the orthogonal paths corresponding to the other (front and back) spiral paths. 

\subsubsection{Root Node Unfolding}
\label{sec:root}
The last internal node of $T$ whose spiral paths are extended to
its parent in this fashion
is the (single) child $b$ of the root band.   As described, these spiral paths
cycle counterclockwise on $b$ to reach the top rim segment
$r$ on the front rim of $b$, which by definition is the back rim of the root band.  Let $\xi_1, \xi_2, \dots, \xi_\ell$
be the extended spiral paths listed in the order encountered in a clockwise walk around the front rim of $b$, 
starting at the top left corner of the root's back face. Here 
$\ell$ is the number of leaf nodes
in $T$.  (For example, $\ell = 7$ in the example from~\autoref{fig:unf0}.)
Let $L(\xi_i)$ be the left endpoint of spiral $\xi_i$ on $r$ and let $R(\xi_i)$ be the right endpoint. 
With this notation, the endpoints from left to right on $r$ are $L(\xi_1),
R(\xi_1), L(\xi_2), R(\xi_2), \dots, L(\xi_\ell), R(\xi_\ell)$.

The next step of the algorithm is to link these $\ell$ spiral paths into a single
path $\xi$ that can be flattened in the plane as a monotone staircase.  
The starting point of $\xi$ is $L(\xi_1)$, and the first part of $\xi$ 
consists of $\xi_1$ followed by $\xi_\ell$.  These two spiral paths
are linked via a connector path on the root band that extends
from endpoint $R(\xi_1)$ to endpoint $R(\xi_\ell)$. 
See the 2D representation of the connector path linking the 
right endpoint of $\xi_1$ to the right endpoint of $\xi_7$ in~\autoref{fig:unf0}. 
This connector path is analogous
to the connector paths followed at the leaf nodes, but here the vertical segment is on the front
face of the root band. 
Because the root node has degree one and its only child is adjacent on its back rim, 
the root's front rim encloses a face of $P$, and so it is possible for
the path to connect in this manner.
This connector path is depicted in~\autoref{fig:root}.
From $R(\xi_\ell)$, the connector path cycles counterclockwise
to reach $R(\xi_1)$. From there, both parts of the path (i.e., the extensions of $R(\xi_\ell)$ and $R(\xi_1$)) 
cycle counterclockwise together towards the front face.
The part of the path extending from $R(\xi_\ell)$ meets
the front face vertical segment at its top endpoint, while the 
other part of the path extending from $R(\xi_1)$
continues cycling to the bottom face and meets the
vertical segment at its bottom endpoint.
(Observe that this path is a mirror image of the one 
depicted in~\autoref{fig:leaf}.)
\begin{figure}[htbp]
\centering
\includegraphics[width=0.6\linewidth]{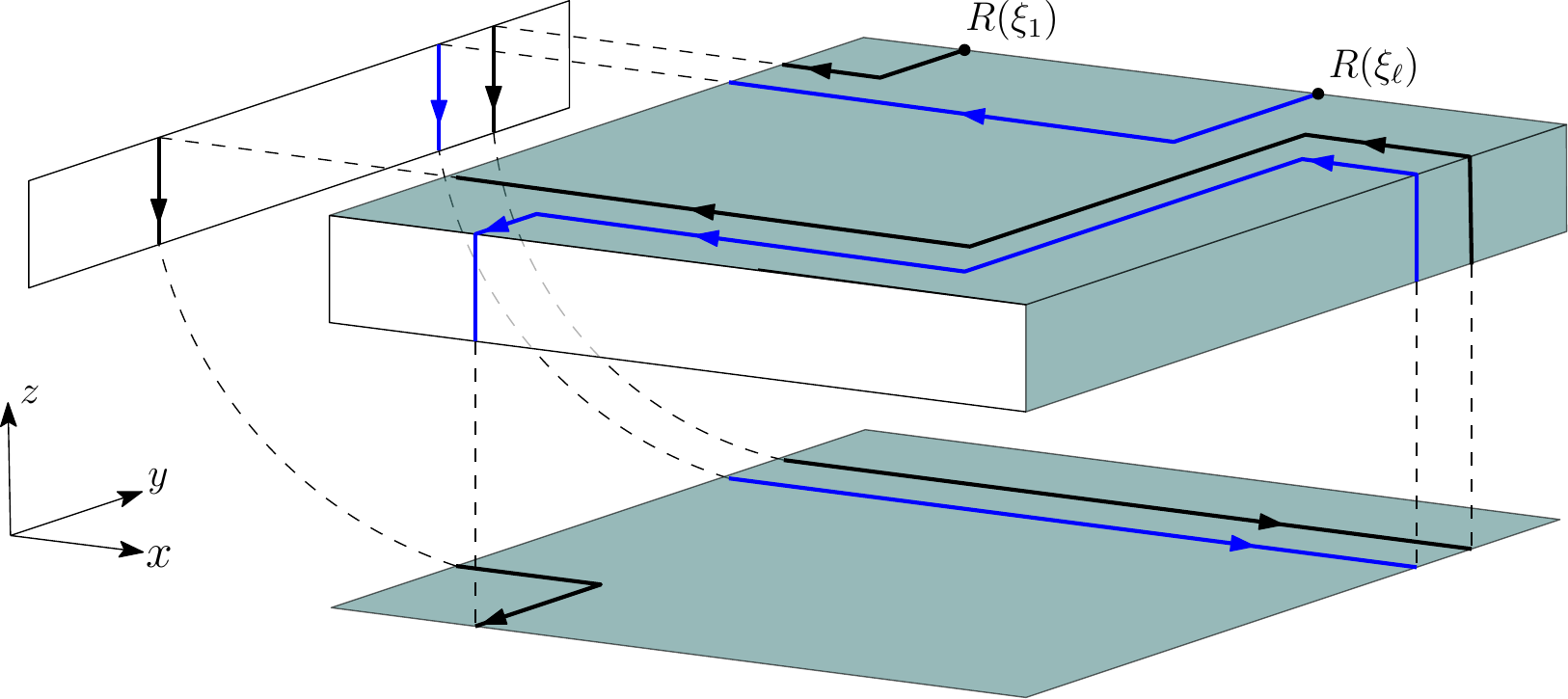}
\caption{The connector path for the spiral paths $\xi_1$ and $\xi_\ell$; arrows indicate the direction followed by the 
unfolding algorithm, up to the vertical segment on the front face. 
} 
\label{fig:root}
\end{figure}
Like a leaf node connector path, 
this path can be laid flat 
in the plane as a monotone staircase. Because the counterclockwise cycling
direction of the connector path is consistent with that of $\xi_1$ and
$\xi_\ell$, $\xi_1$ can be laid flat
on one side of the path
and $\xi_\ell$ can be laid flat on the other side, thus forming
a single monotone staircase. In this way
we link $\xi_1$ and $\xi_\ell$ to form the first part of $\xi$.

Continuing to link the spiral paths to form $\xi$, $L(\xi_\ell)$ is linked  
to endpoint $L(\xi_2)$ using a connector path that runs alongside the previous connector path.  Similarly, this connector path flattens to a monotone 
staircase and connects the two flattened staircases $\xi_\ell$ and $\xi_2$. 
The remaining endpoints are paired up similarly and linked with
connector paths. Specifically, $R(\xi_i)$ is linked to $R(\xi_{\ell - i + 1})$
and  $L(\xi_i)$ is linked to $L(\xi_{\ell - i + 2})$), for $i = 2, 3, \dots$ until one unpaired endpoint remains: $L(\xi_{\ell / 2 + 1 })$ (if $\ell$ is even)
or $R(\xi_{(\ell+1)/2})$ (if $\ell$ is odd). 
This endpoint is where $\xi$ terminates. See $R(\xi_4)$ in~\autoref{fig:unf0}. 
 
\subsubsection{Completing the Unfolding}
To complete the unfolding of $P$, the spiral $\xi$ is thickened in the $+y$ and $-y$ direction so that it completely covers
each band. This results in a thicker strip, which can be unfolded as a staircase in the plane. Then the forward and rearward faces of $P$ are 
partitioned by imagining the band's top rim edges illuminating downward light rays in these faces. The illuminated 
pieces are then ``hung" above and below the thickened staircase, along the corresponding illuminating rim 
segments that lie along the horizontal edges of the staircase.

\subsection{Unfolding Genus-0 Orthogonal Polyhedra}
\label{sec:unf0}
The unfolding algorithm described in~\autoref{sec:extrusions} for extrusions generalizes to 
all genus-0 orthogonal polyhedra as described in~\cite{Chang2015}, 
so we briefly present the main ideas here and refer the reader to~\cite{Chang2015} for details. 

Instead of partitioning $P$ into boxes, the unfolding algorithm
partitions $P$ into slabs as defined in Section~\ref{sec:intro}.
It then creates an unfolding tree $T$ in which each node corresponds to either an outer band (surrounding a slab)
or an inner band (surrounding a hole). 
Each edge in $T$ corresponds to a $z$-beam, which is a thin vertical rectangular strip from a frontward or 
rearward face of $P$ connecting a parent's rim to a child's rim. Note that a $z$-beam may have zero geometric height, 
when two rims share a common segment. 
The spiral paths connect vertically along 
the $z$-beams when transitioning from a child band to its parent. 
For a parent band $b$, its front (back) children are those
whose $z$-beams connect to $b$'s front (back) rim.
It was established in~\cite{Damian-Flatland-O'Rourke-2007-epsilon} that the back 
rim of each leaf node in $T$ encloses a face of $P$.  

\subsubsection{Assigning Unfolding Directions}
\label{sec:dir}
Unlike the case of protrusions, where all bands are unfolded in the same direction (i.e., either all counterclockwise or all clockwise), 
general genus-0 orthogonal polyhedra may require different unfolding directions for different bands. For example, if a $z$-beam 
is incident to a top rim edge of the parent and a bottom rim edge of the child, then the unfolding direction (viewed from $y=-\infty$) changes when 
transitioning from the child to the parent.~\autoref{fig:g0dir}a shows such an example.

\begin{figure}[htbp]
\centering
\includegraphics[width=0.8\linewidth]{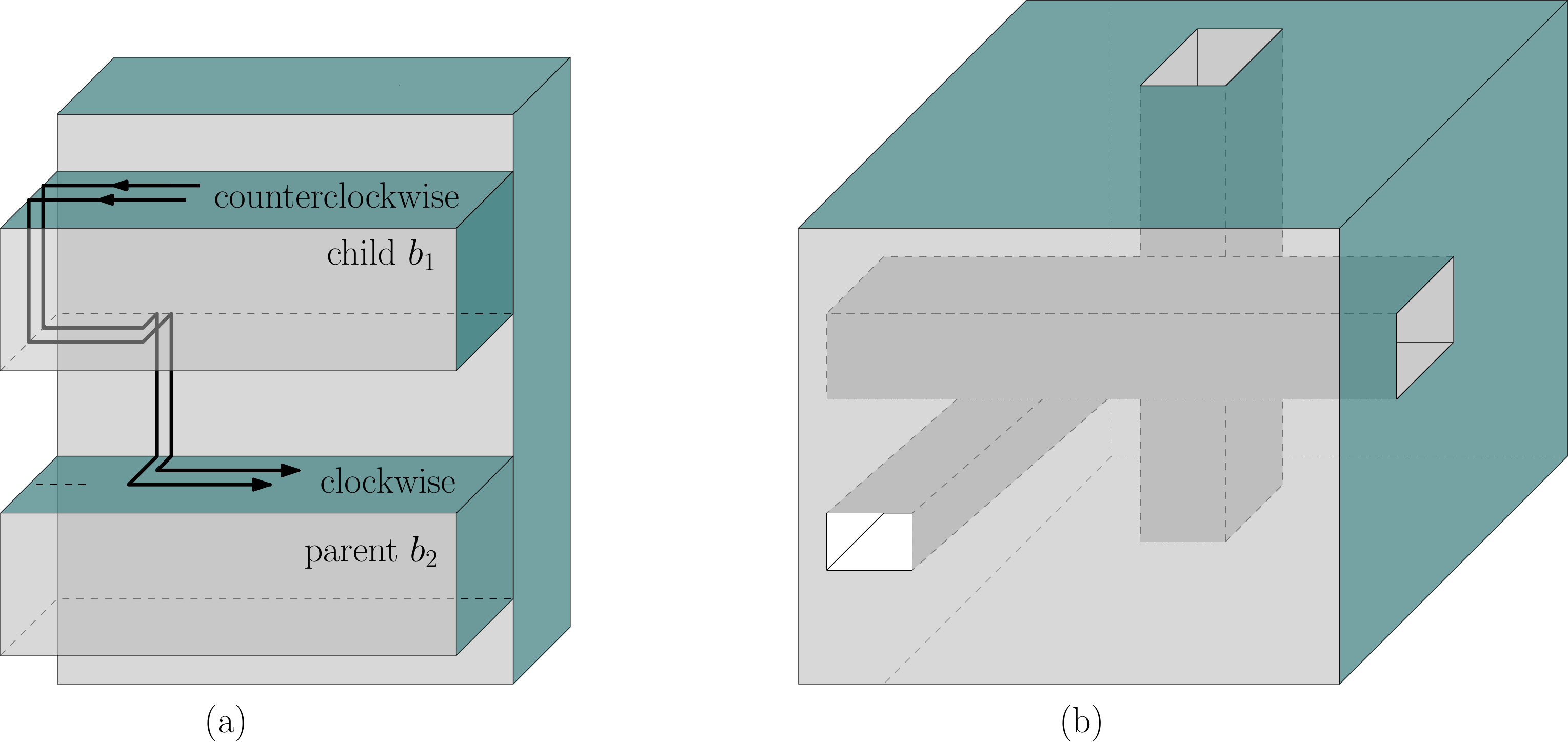}
\caption{(a) The unfolding direction changes when transitioning from child $b_1$ to parent $b_2$ 
(b) \autoref{lem:root}: a polyhedron of genus $g = 3$ with no slicing direction that yields a face-node.} 
\label{fig:g0dir}
\end{figure}

We assign unfolding directions for each band in a preorder traversal of the unfolding tree $T$. 
Set the unfolding direction for the root band to counterclockwise. At each band node $b$ visited in a preorder traversal of $T$, 
if the edge in $T$ connecting $b$ to its parent corresponds to a $z$-beam incident to both top and bottom rim points, then set 
the unfolding direction for $b$ to be opposite to the one for the parent (i.e., if the unfolding direction for the parent is counterclockwise, 
then the unfolding direction for $b$ will be clockwise, and vice versa.)
Otherwise, the endpoints of the $z$-beam connecting 
$b$ to its parent are both top rim points or both bottom rim points,  and in that case $b$ inherits the unfolding direction of its parent.
(Recall that a $z$-beam can have zero geometric height where two rims overlap.)

\subsubsection{The Unfolding Procedure}
The  unfolding of a leaf band $b$ 
follows the description 
in~\autoref{sec:leaf} (Figure~\ref{fig:unf0}), except that the unfolding proceeds in the direction
$d$ assigned to $b$ (as described in~\autoref{sec:dir}), and the spiral path may cycle
around multiple band faces instead of just four.  When the two endpoints reach
the $z$-beam on $b$'s front rim, they 
track vertically along the $z$-beam, stopping side-by-side on the rim of $b$'s parent. 

At each internal node $b$ in $T$, 
the unfolding 
proceeds as described in~\autoref{sec:internal}. 
Observe that there is a natural cyclic ordering of $b$'s front (back) children that is determined by their
$z$-beam connections around $b$'s front (back) rim, which guides the order in which we process $b$'s children. 
Once the pairs of endpoints reach the
$z$-beam connection to $b$'s parent on $b$'s front rim, they move vertically along the $z$-beam, stopping on the rim of $b$'s parent.
At the root node, these strips 
are glued together as described in~\autoref{sec:root} (with the notion of ``left'' and ``right'' altered to match  
the cyclic ordering of the children, so that $L(\xi)$ and $R(\xi)$ are always encountered in this order in a clockwise walk along 
the rim).
In addition, the spiral paths followed on inner bands
are the same as those described previously for outer bands.
For example, assume that the inner band $b_{12}$ that forms a dent in the example from~\autoref{fig:defs1}
is a leaf band in $T$. 
Note that the interior of $P$ surrounds $b_{12}$ on all sides, except for the front
which is the entrance to the dent.
Then the connector path is the same as in ~\autoref{fig:leaf}a but is now viewed as
cycling on the surface of $P$ inside the dent. 

In the unfolded staircase, the portion
of $\xi$ on a $z$-beam corresponds to a vertical riser. Thickening $\xi$ is 
proceeds as in the case of 
extrusions. The partitioning of the
forward and rearward faces also follows 
the case of extrusions, but
in addition to shooting illuminating rays down from top rim edges, bottom rim edges
also shoot rays downward to illuminate portions of faces not illuminated by
the top edges. 
The face pieces resulting from this partitioning method
are hung above and below the staircase, as in the case of extrusions,
as described in~\cite{Damian-Flatland-O'Rourke-2007-epsilon}. 

\section{Unfolding Genus-2 Orthogonal Polyhedra}
\label{sec:genus2}
The unfolding algorithm described in~\autoref{sec:genus0} depends on two key properties
of $P$ that are not necessarily true if $P$ 
has genus $1$ or $2$.
First, it requires the existence of a band with
a rim enclosing 
a face of $P$ that can serve as the root node of $T$. And second, it requires
the back rim of each leaf node in $T$ to enclose 
a face of $P$. These two requirements are needed so that the connector paths can 
use vertical strips on the enclosed faces in the unfolding.
As a simple example of a genus-1 polyhedron for which neither property
holds, 
consider the case when $P$ is a box with a $y$-parallel
hole through its middle. Slicing $P$ with $y$-perpendicular planes
results in a single slab having one outer and one inner band.
In this case, no rim encloses a face of $P$.  If we rotate
$P$ so that the hole is parallel to the 
$x$ axis instead, slicing produces four bands, and the band surrounding the
frontward (or rearward) box could
serve as the root node.  But every 
unfolding tree for the four bands contains a leaf whose back rim doesn't enclose a face of $P$. 

In this section, we first show that there always exists 
an orientation for $P$ such that at least
one band has a rim enclosing a face of $P$ that can be used for the root of $T$.  
Then we describe an algorithm that computes 
 an unfolding tree for which we can prove that the number of 
leaf bands whose back rims do not enclose a face
of $P$ is at most $g$,  where $g$ is the genus of $P$. 
Finally, we describe changes to the unfolding algorithm
that allow it to handle up to $g$ leaves that don't enclose a face of $P$, for $g \leq 2$.

\subsection{The Rim Unfolding Tree $T_r$}
\label{sec:tree2}
In order to establish these new results, we need finer-grained structures
than the band-based $G$ and $T$, which we call $G_r$ and $T_r$,
both of which are rim-based.
We define the \emph{rim graph} $G_r$ for $P$ in the following way.
For each band $b$ of $P$, add two nodes $r_b$ and $r'_b$ to $G_r$
corresponding to each of $b$'s rims.
Add an edge connecting $r_b$ and $r'_b$ and call it a \emph{band edge}, or a $b$-\emph{edge} for short. 
For each pair of rims that can be connected by a $z$-beam, add an edge between them
in $G_r$, and call it a $z$-\emph{beam edge}, or a $z$-\emph{edge} for short. 
When referring to $G_r$, we will use the terms node and rim interchangably.
For any simple \emph{cycle} $C$ in $G_r$, we distinguish between \emph{$b$-\emph{nodes}} of $C$, 
which are endpoints of $b$-edges in $C$, and $z$-nodes of $C$, incident to two adjacent 
$z$-edges in $C$.
For any subgraph $J \subseteq G_r$, we use $V(J)$ and $E(J)$ to denote the set of nodes  
and the set of edges in $J$, respectively. 

Call a rim of $G_r$ that encloses a face of $P$ 
a \emph{face-node}.
A \emph{nonface-node} in $G_r$ is 
a node whose rim does not enclose a face of $P$.

\begin{proposition}
A rim $r$ is a face-node of $G_r$ 
if and only if every $z$-beam extending from a horizontal edge of $r$ and going up or down on the surface of $P$ hits $r$. 
A face-node of $G_r$ is necessarily of degree one. 
\label{prop:enclosedFace}
\end{proposition}

\begin{lemma}
If polyhedron $P$ has genus $g \le 2$, 
then there is a direction for slicing $P$ such that
$G_r$ includes a face-node $r_F$.
\label{lem:root}
\end{lemma}
\begin{proof}
Define the \emph{extreme faces} of $P$ as those faces flush with the smallest bounding box
enclosing $P$.
There must be at least one extreme face in each of the six directions $d \in \{\pm x, \pm y, \pm z\}$.
If any extreme face $F$, say in direction $d$, is simply connected, then slicing $P$
with a $d$-plane (parallel to $F$) just adjacent to $F$ will create a band $b$ one of whose rims
$r_F$ encloses $F$. Thus, slicing $P$ with $d$-planes will result in $G_r$ including the face-node $r_F$,
and the lemma is established.

Assume henceforth that each of the at least six extreme faces of $P$ is not simply connected.
So each extreme face $F$ includes at least one inner band $b$.
We now classify these bands $b$ into two types.

Let $r$ be the rim of an inner band $b$ in face $F$.
If cutting along $r$ separates the surface of $P$ into two pieces,
$P'$ which includes $b$, and the remainder $P \setminus P'$, then
we say that $b$ is a \emph{cave-band}.
Let $M$ be the ``mouth'' of the cave: the portion of the $Y$-plane enclosed by the rim $r$.
$P'$ is a ``cave'' in the sense that an exterior path that enters through $M$ can only exit
$P'$ back through $M$ again.
A band $b$ that is not a cave-band is a \emph{hole-band}.
These have the property that there is an exterior topological circle that passes through 
the mouth $M$ once, and so exits $P'$ elsewhere.

Let $P'$ be a cave with mouth $M$.
$P' \cup M$ is an orthogonal polyhedron $P'_M$, inverting
what was exterior to $P$ to become interior to $P'_M$.
Say that cave $P'$ has genus $0$ if $P'_M$ has genus $0$.
We now claim that the lemma is satisfied if $P$ has a genus-$0$ cave.
For we may apply the same procedure to $P'_M$: 
Examine its extreme faces (one of which is $M$).
If any extreme face (other than $M$) is simply connected, we are finished.
Otherwise, each extreme face includes an inner band $b'$.
It cannot be the case that $b'$ is a hole band, for then $P'_M$ has
genus greater than $1$.
Moreover, $b'$ cannot be a cave band for a cave of genus greater than $1$.
For in both cases, we could cut a cycle on the surface of $P'$ that would not disconnect 
$P'$. So $b'$ must determine a genus-$0$ cave band, and the argument repeats.
Eventually we reach a simply connected extreme face.

Now we have reduced to the situation that each of $P$'s six or more
extreme faces contains either a hole-band, or a genus-$(\ge 1)$ cave band.
Let the number of these bands be $h$ and $c$ respectively.
We now account for the  the genus $g$ of $P$, and the number of extreme faces.
Each cave of genus-$(\ge 1)$ contributes at least $1$ to $g$.
A hole-band in an extreme face could exit through that same face, 
or exit through a different extreme face,
or exit through a non-extreme face.
In the first two cases, two hole-bands contribute $1$ to $g$;
in the third case, one hole band contributes $1$ to $g$.
So $h$ hole-bands contribute at least $h/2$ to the genus,
and we have the inequality $c + h/2 \le g \le 2$.

We have defined $h$ and $c$ to be the number of such bands in extreme faces,
and we know that each of the at least six extreme faces must have one or more
hole- or genus-$(\ge 1)$ cave-bands. So we must have $c + h \ge 6$.
But these two inequalities have no solutions in non-negative integers. 
{\hfill\qed} 
\end{proof}


\noindent
\autoref{fig:g0dir}b shows that \autoref{lem:root} is tight.
Henceforth we assume $P$ is oriented so that the  direction
guaranteed by the lemma slices $P$ with $Y$-planes, and so $G_r$ has
a face-node $r_F$. This node will become the root of the unfolding tree.

\begin{figure}[htpb]
\centering
\includegraphics[width=\linewidth]{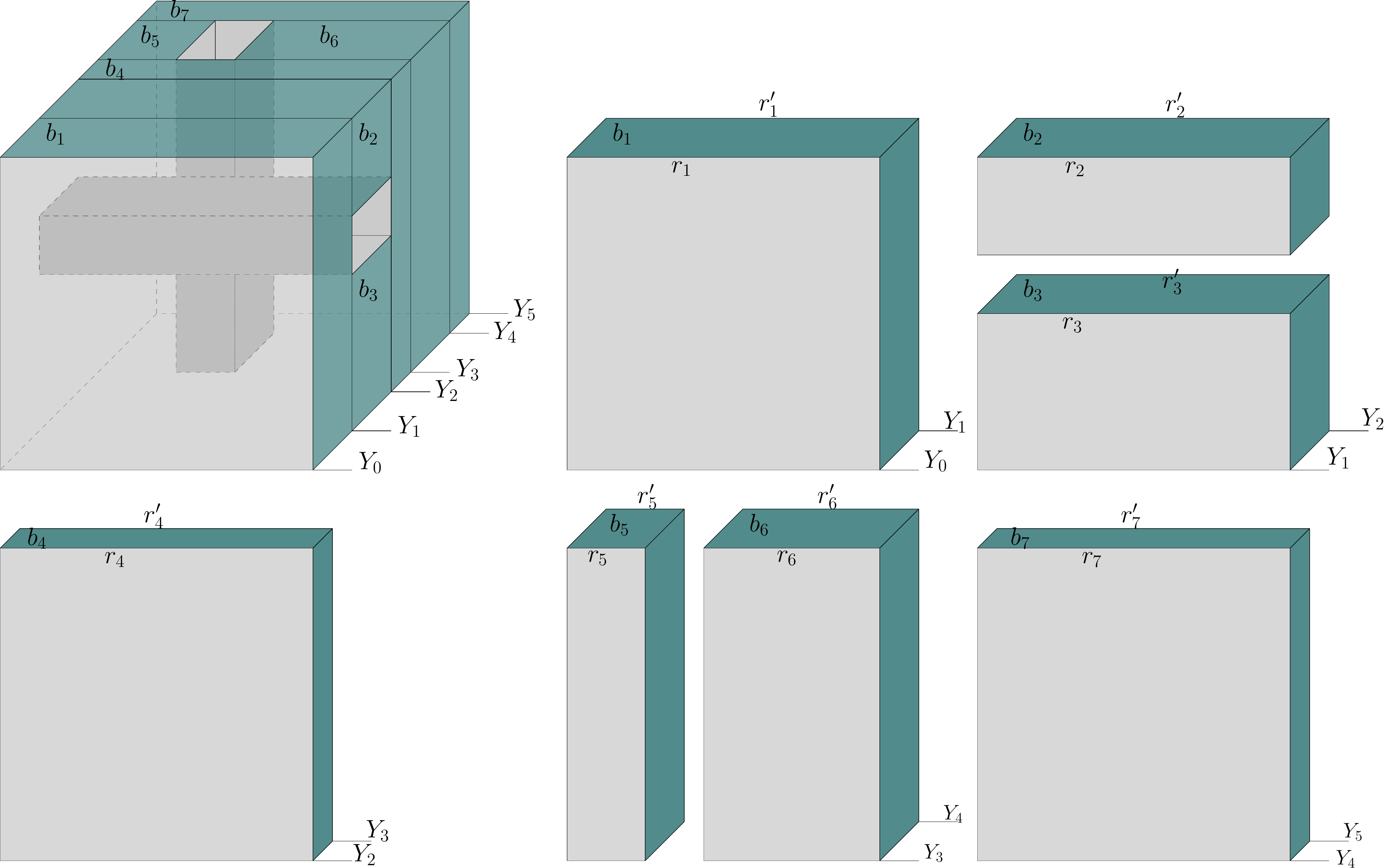}
\caption{Genus-2 polyhedron partitioned into slabs.} 
\label{fig:g2poly}
\end{figure}

\noindent
\autoref{fig:g2poly} shows an example of a genus-2 polyhedron sliced with $Y$-parallel planes in the direction 
identified by~\autoref{lem:root}, which yields two face-nodes $b_1$ and $b_7$.
The rim graph $G_r$ for this polyhedron is depicted (ahead)  
in~\autoref{fig:g2tree}b, with $r_1$ selected as root.

\begin{lemma}
$G_r$ is connected and contains no nonface-leaf nodes. 
\label{lem:terminal}
\end{lemma}
\begin{proof}
Call a maximal connected surface piece of $P$ located in a plane $Y_i$ a $z$-\emph{patch} (suggestive of the fact that 
it might contain $z$-beams). 
First note that the subset of rims belonging to a $z$-patch induce a connected component in $G_r$ 
(call it a $z$-\emph{patch component}) that contains $z$-edges only. 
The connectedness of $P$'s surface implies that all $z$-patches are connected together by bands. In $G_r$, this 
corresponds to all $z$-patch components being connected together by $b$-edges. It follows that $G_r$ is connected.

Next we show that there are no nonface-leaves in $G_r$.
Suppose there is a node $r$ in $G_r$ of degree $1$ that does
not enclose a face of $P$.  Because all nodes in $G_r$ are
connected by a $b$-edge to the rim on the other side of the
band, there is a $b$-edge adjacent
to $r$.
Now consider extending $z$-beams up and down from every horizontal
edge of $r$. Because $r$ does not enclose a face of $P$, 
at least one of the $z$-beams must hit the rim of another band
by~\autoref{prop:enclosedFace}. 
But then $r$ also has a $z$-edge adjacent to it,
giving it a degree of at least $2$, a contradiction.
{\hfill\qed} 
\end{proof}
Note that $G_r$ for the genus-$3$ example in~\autoref{fig:g0dir}b has no leaf nodes at all,
and so satisfies this lemma vacuously.

Our next goal is to find a rim spanning tree $T_r$ of $G_r$ with at most $g$ nonface-leaves,
which will ultimately similarly limit the number of nonface-leaf nodes of $T$. 
The {\sc RimUnfoldingTree} method for 
achieving such a $T_r$ is outlined in~\autoref{alg:unftree}. 
It reduces $G_r$ to a tree by repeatedly removing a $z$-edge 
from an existing cycle, thus breaking the cycle.  In addition, it does this 
in such a way that at most $g$ nonface-leaf nodes are created. 
If we were to break a cycle by removing an 
arbitrary $z$-edge 
from it, it may be that both endpoints of the $z$-edge 
have degree two, and thus removing it would result in the creation 
of two new leaf nodes, both of which would be nonface-nodes. 
To avoid this, our {\sc RimUnfoldingTree} algorithm 
strategically selects a $z$-edge $e$ with at least one endpoint, say $u$, of degree 
$3$ or more. Thus the removal of $e$ results in the creation of at most one
new leaf. More importantly, the algorithm ensures that one of $u$'s three or more adjacent
edges is an edge
that is not part of any current (simple) cycle. Call this edge 
$e'$, and note that $e'$ will never be removed by the algorithm, because the algorithm 
only removes cycle edges.  
The existence of $e'$ guarantees that $u$'s degree will not drop below $2$
(for if the degree of $u$ were to reach $2$, then because one of the two adjacent
edges is $e'$ and not part of a cycle, the other adjacent edge cannot be part of a cycle either, and 
therefore neither edge will be removed). 
This property of $u$ will be important in bounding the number of nonface-leaf nodes created by the algorithm.

\begin{algorithm}
\centerline{$T_r$ = {\sc RimUnfoldingTree}($G_r$)}
{\hrule width 0.92\linewidth}\vspace{0.8em}
Initialize $T_r \leftarrow G_r$ \\
\While{$T_r$ is not a tree}{
	Let $H \subset T_r$ be the subgraph of $T_r$ induced by all simple cycles in $T_r$ \\
	Pick an arbitrary node $u \in V(H)$ incident to an edge $e'$ in $E(T_r) \setminus E(H)$ \\ 
      	Pick an arbitrary $z$-edge $e \in E(H)$ incident to $u$ \\
	Remove $e$ from $T_r$
 }
 \Return $T_r$
\vspace{2mm}\noindent{\hrule width 0.92\linewidth} \vspace{1mm} 
\caption{Computing a rim spanning tree of $G_r$.}
\label{alg:unftree}
\end{algorithm}

\begin{lemma}
The {\sc RimUnfoldingTree} algorithm produces a spanning tree of $G_r$.
\end{lemma}
\begin{proof}
By~\autoref{lem:root}, $G_r$ (and therefore $T_r$) includes a node of degree one that is not part of a cycle in $G_r$ and therefore is not in $H$. This implies that there is at least one edge in $E(T_r) \setminus E(H)$ incident to a node $u$ of $H$ (because $T_r$ is connected). This proves the existence of the node $u$ picked in each iteration of the {\sc RimUnfoldingTree} algorithm. Because $u$ is part of a least one cycle in $H$, its degree is at least two in $H$. The edge in $E(T_r) \setminus E(H)$ incident to $u$ contributes another unit to the degree of $u$; therefore $u$ has degree at least three in $T_r$. 
By the definition of $G_r$, no two $b$-edges in $G_r$ are adjacent, since any two $b$-edges are connected by a path of one or more $z$-edges in $G_r$. 
(Recall a $z$-edge might have zero geometric height, when two rims share a common segment.) 
This implies that, out of the two or more edges in $E(H)$ incident to $u$, at least one is a $z$-edge. This proves the existence of the edge $e$ picked in each iteration of the {\sc RimUnfoldingTree} algorithm. Removing $e$ from $T_r$ breaks at least one cycle in $H$, so the size of $H$ decreases in each loop iteration. It follows that the {\sc RimUnfoldingTree} algorithm terminates and produces a tree $T_r$ that spans all nodes of $G_r$.
{\hfill\qed} 
\end{proof}

\begin{theorem}
The number of nonface-leaves in the rim tree $T_r$ produced by the {\sc RimUnfoldingTree} algorithm is no greater than the genus $g$ of $P$. 
\label{thm:termleaves}
\end{theorem}
\begin{proof}
Consider a $z$-edge $e = (u, v)$ removed from $T_r$ in one iteration of the {\sc RimUnfoldingTree} algorithm. The node $u$ is incident to at least two edges in $E(H)$ and at least one edge in $E(T_r) \setminus E(H)$; therefore its degree is at least three in $T_r$. The removal of $e$ from $T_r$ leaves $u$ of degree at least two, so $u$ does not become a leaf in $T_r$. This argument holds even if $u$ is picked repeatedly in 
future iterations of the {\sc RimUnfoldingTree} algorithm, so $u$ will not become a leaf in $T_r$. 

If $v$ has degree three or more in $T_r$ prior to removing $e$ from $T_r$, then $v$ does not become a leaf after removing $e$. So suppose that $v$ has degree two in $T_r$ before removing $e$. Recall that every node in $G_r$ is connected by a $b$-edge to the other rim of its band. Because the 
{\sc RimUnfoldingTree} algorithm never removes a $b$-edge, this property holds in $T_r$ as well. It follows that 
 the other edge incident to $v$ (in addition to the $z$-edge $e$) must be a $b$-edge. Hence 
any leaf node in $T_r$ created by the {\sc RimUnfoldingTree} algorithm is an endpoint of a $b$-edge in $T_r$. 

By~\autoref{lem:terminal}, $G_r$ does not include any nonface-leaves, so any nonface-leaves in $T_r$ must have been created by the {\sc RimUnfoldingTree} algorithm. 
Let $r_1, r_2, \ldots, r_k$ be the set of leaves in $T_r$ created by the {\sc RimUnfoldingTree} algorithm. 
If $k \le g$, then the theorem is true. 

\begin{figure}[htbp]
\centering
\includegraphics[width=0.9\linewidth]{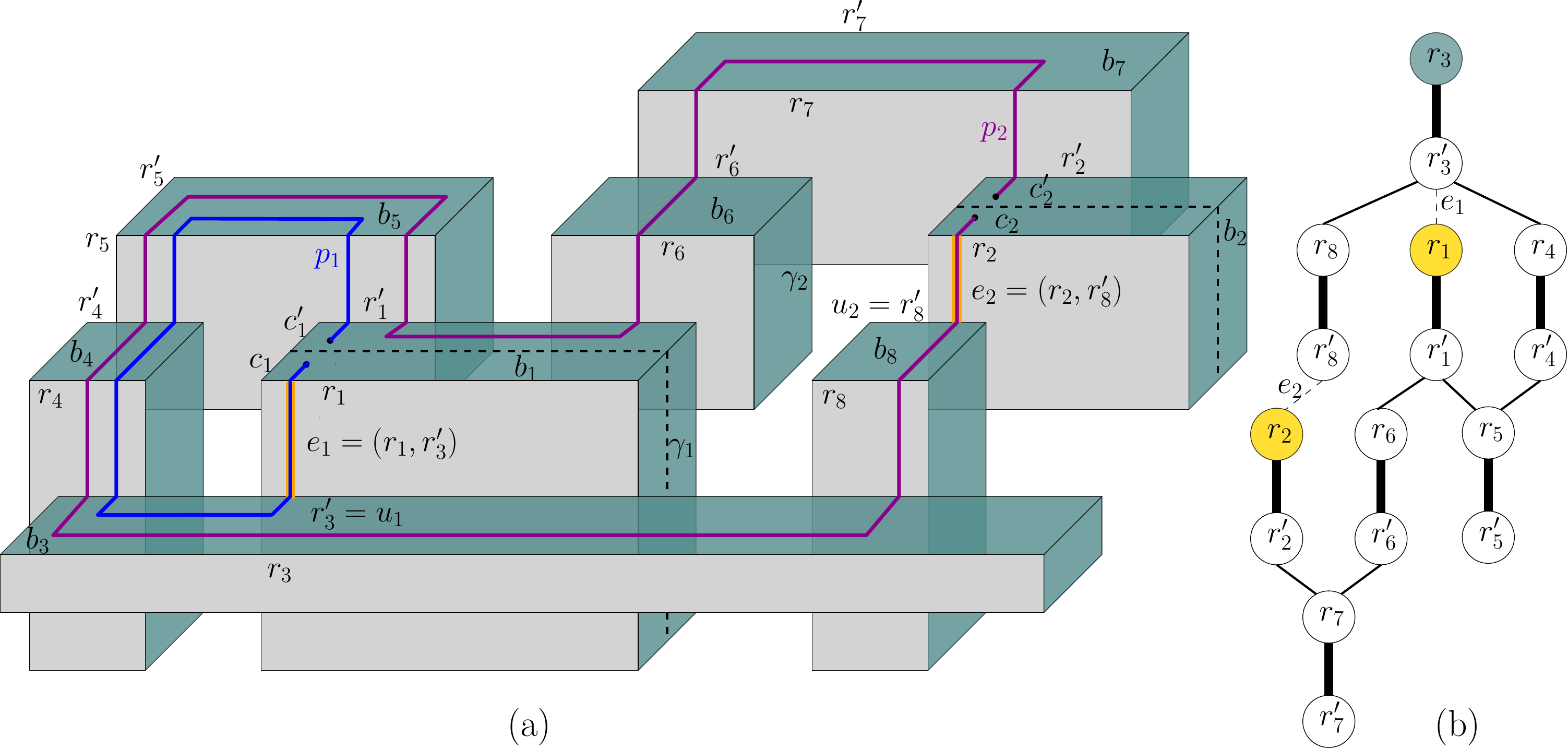}
\caption{\autoref{thm:termleaves} (a) example polyhedron of genus 2; cuts around the middle of bands $b_1$ and $b_2$ are marked with dashed lines  (b) 
Rim tree $T_r$ with root $r_3$ and $2$ nonface-leaves, $r_1$ and $r_2$.
} 
\label{fig:genus}
\end{figure}

Assume now that $k > g$. For each $i$, let $b_i$ be the band with rim $r_i$. 
Because $T_r$ includes all the $b$-edges from $G_r$, it must be that $T_r$ includes the $b$-edge $(r_i, r'_i)$ corresponding to $b_i$, so the parent $r'_i$ of $r_i$ in $T_r$ 
is the other rim of $b_i$.
For each $i = 1, 2, \ldots, k$, we perform a cut on $P$'s surface 
along a closed curve around the middle of $b_i$, between its rims $r_i$ and $r'_i$. 
Refer to~\autoref{fig:genus}a, which shows the cuts around the middle of $b_1$ and $b_2$ as dashed lines. Note that $r_1$ and $r_2$ are nonface-nodes corresponding to $b_1$ and $b_2$ respectively, as inferred from the rim unfolding tree $T_r$ shown in~\autoref{fig:genus}b (edges of $T_r$ are marked solid, with $b$-edges 
thicker than $z$-edges).
We now show that these cuts do not disconnect the surface 
of $P$, contradicting our assumption that $k > g$. 

To see this, consider any two points $c_i$ and $c'_i$ on $b_i$ that are separated by the cut around $b_i$, with $c_i$ on the same side of the cut 
as $r_i$, and $c'_i$ on the same side of the cut as $r'_i$. 
Let $e_i = (r_i, u_i)$ be the $z$-edge whose removal from $T_r$ created the leaf $r_i$ (so $u_i$ here plays the role of $u$ in 
the {\sc RimUnfoldingTree} algorithm, and by our observation above, $u_i$ is not a leaf in $T_r$). 
We construct a path $p_i$ on the surface of $P$ connecting $c_i$ to $c'_i$ as follows: 
the path $p_i$ starts at $c_i$, moves towards $r_i$ and along the $z$-beam corresponding to $e_i$ to the rim $u_i$, 
then follows the path in $T_r$ from $u_i$ to the rim $r'_i$, and finally across $b_i$ to $c'_i$. 
See~\autoref{fig:genus}, which traces the paths $p_1$ and $p_2$ on the polyhedron surface.
Note that $r_i$ is the only leaf visited by $p_i$ (because $u_i$ is not a leaf), so $p_i$ does not cut across any of the 
leaves $r_j$, for $j \neq i$. This implies that $p_i$ connects $c_i$ and $c'_i$ in the presence of all the other cuts 
$p_j$, with $j \neq i$. Since this is true for each $i$, we conclude that these cuts leave the surface of $P$ connected, 
contradicting our assumption that $k > g$. 
{\hfill\qed} 
\end{proof}


\subsection{The Unfolding Algorithm}
\label{sec:unf2}
Let $T_r$ be the rim tree computed by the {\sc RimUnfoldingTree} algorithm described in~\autoref{sec:tree2}. We pick the root of $T_r$ to be a face-node  identified by~\autoref{lem:root}, and call its corresponding band the \emph{root band}. This guarantees that the front face of the root band
is a face of $P$ and thus it can be 
used in constructing the connector paths linking the spiral
paths associated with the root's children. 

For ease in describing the modified spiral paths 
we use for genus-$1$ and genus-$2$ polyhedra, 
we first 
convert $T_r$ into a standard unfolding tree $T$ having bands for nodes rather than rims.
We do this by simply contracting the $b$-edges.  Specifically, we replace each $b$-edge
$(r_i,r_i')$ 
and its two incident nodes $r_i$ and $r'_i$ by 
a single node $b_i$ whose incident edges 
are the $z$-edges incident to $r_i$ or $r_i'$. Let $T$ be the tree resulting 
from $T_r$ after 
contracting all the $b$-edges, with its root node corresponding to the root band. 

Observe that the edges in $T$ are in a one-to-one correspondence
with the $z$-edges in $T_r$. 
For example, consider a $z$-edge $(r_i,r_j)$ in $T_r$ such that node $r_i$ is the parent of $r_j$.
Then in $T$, there is an edge $(b_i,b_j)$ where $b_i$ is the parent of $b_j$.  Furthermore, 
$r_j$ is the front rim of $b_j$ because of the $z$-beam connection
to its parent determined by the $z$-edge $(r_i,r_j)$.   
Any other nodes adjacent to $b_j$ in $T$
are  its front or back children, depending on whether the
corresponding $z$-edges in $T_r$ connect to its front rim ($r_j$) or its back rim.
In this way, the rim connections that are explicitly represented in $T_r$ are preserved in
$T$ through the assignment of front and back rims/children. 

There is not, however, an immediate one-to-one correspondence
between leaf nodes in $T_r$ and leaf nodes in $T$. 
If $r_i'$ is a leaf in $T_r$ and its parent  $r_i$
has no other children,  then clearly $b_i$ in $T$ has degree
$1$, and we note that its back rim is $r_i'$.  
See, for example, the leaf $r_5$ in~\autoref{fig:g2tree}c (ahead) and the corresponding leaf 
$b_5$ in~\autoref{fig:g2tree}d.
But suppose $r_i$ has one or more other children in $T_r$ besides the leaf $r_i'$. Then these other
children are connected via $z$-edges to $r_i$, and node 
$b_i$ in $T$
has degree 
greater than $1$ and is not a leaf.  Specifically,
$b_i$ in $T$ has one or more front children connected via $z$-beams to its
front rim $r_i$, and it has no back children (because its back rim $r_i'$ is a leaf in $T_r$).  
This is the case for the leaf node $r_2$ from~\autoref{fig:g2tree}c, whose parent 
node $r'_2$ has another child $r_4$; the corresponding node $b_2$ in $T$ is not a leaf in $T$.  
Similarly, the parent $r_7$ of leaf node $r'_7$ from~\autoref{fig:g2tree}c has another child $r'_5$, 
and the corresponding node $b_7$ is not a leaf in $T$.
But we handle this in the same way as in a standard 
unfolding tree (as described in~\autoref{sec:extrusions-unftree})
by splitting band $b_i$ into two bands $b_i'$ and $b_i''$.  
In $T$, $b_i$ is
replaced by $b_i'$ and has $b_i''$ as a back leaf child.  
See~\autoref{fig:g2tree}d, which shows node $b_2$ split into $b'_2$ and $b''_2$, and 
node $b_7$ split into $b'_7$ and $b''_7$.
The front
rim of $b_i'$ is rim $r_i$ and the back rim of $b_i''$ is rim $r_i'$.
In this way,
each leaf $r_i'$ in $T_r$ has a corresponding leaf node in $T$ whose
back rim is $r_i'$, and vice versa.  If $r_i'$ is a nonface-leaf in $T_r$, then we  will
also say the corresponding leaf
in $T$ is a nonface-leaf, meaning that its band's back rim does not enclose
a face of $P$. These observations, along with~\autoref{thm:termleaves}, establish the following corollary.  

\begin{corollary}
The number of nonface-leaves in the unfolding tree $T$ is at most $g$, where $g \ge 0$ is the genus of $P$.
\label{cor:termleaves}
\end{corollary}

Using $T$, we assign unfolding directions to each band as described in~\autoref{sec:dir}. 
If $T$ has no nonface-leaves, then 
we complete the unfolding of $P$ using the linear unfolding algorithm from~\cite{Chang2015} (as summarized in~\autoref{sec:genus0}). 
We now show how to modify the unfolding algorithm to handle the cases when $T$ has
one or two nonface-leaves.   
Because a nonface-leaf $b$ has a back rim that does not enclose a face of $P$,  
there might be no vertical 
back face segment available for $b$'s connector path.
For these leaves, we use a spiral path that has one
endpoint on $b$'s back rim and the other endpoint on $b$'s 
front rim. The path cycles 
around the band faces in the unfolding direction of $b$
from the back point to the front
point. (In~\autoref{fig:leaf}, this would be just the 
portion of the path 
that extends from the top of the back rim to endpoint $e_1$.) 
\autoref{cor:termleaves} implies that this can 
occur for at most two leaves in $T$ (since we assume $P$ to have genus 
$g \le 2$). 
For all face-leaves, we 
proceed as before in extending both ends of the spiral path from the band's
back rim to its front rim.

The processing of internal nodes is handled 
as before by extending
the spiral paths towards the root.
The only difference is that for the (at most $2$) spiral paths
originating at nonface-leaves, there is only one end
of the path to extend. 
After processing the internal nodes, the ends of the spiral paths 
are at the root band.   Specifically, each spiral path originating at a face-leaf has two endpoints
located on the back rim of the root band, and each spiral path originating at a nonface-leaf has one endpoint located on the back rim of the root band. 
The challenge here is  to connect all  these spiral paths together at the root band into a single final strip that starts on the back rim of one nonface-leaf and (if there is a second nonface-leaf) ends on the back rim of the other.  
(This is one place where the assumption that $g \le 2$, and so there are at most
two nonface-leaves (by~\autoref{cor:termleaves}), is crucial.)
Because the front rim of the root node does enclose 
a face of $P$, it is possible to use strips from its enclosed face for the connectors. 

\subsubsection{Root Node Unfolding: One Nonface-Leaf Case}
\label{sec:oneterminal}
First we describe how to link the leaves' extended spiral paths together at the root node, and unfold the root band in the process,  for the case when $T$ has exactly one nonface-leaf. 
Let $\xi_1, \xi_2, ..., \xi_\ell$ be the spiral paths corresponding to the $\ell$ face-leaves in $T$ 
(excluding the one nonface-leaf). 
After processing all the internal nodes, the two ends of each of these spiral paths are located side-by-side on the back rim $r$ of the root band,
as previously illustrated in~\autoref{fig:unf0}. 
Let $t$ be the spiral path corresponding to the nonface-leaf. One end of $t$ is on $r$, and the other end is on the back rim of its leaf. If $\ell > 1$, we assume the spiral paths of the face-leaves are labeled in clockwise order around the root's back rim, with $t$ located in the middle between $\xi_{\lceil \ell/2 \rceil}$ and $\xi_{\lceil \ell/2 \rceil + 1}$. We then begin by linking all these spiral paths into one strip $\xi$ as described in~\autoref{sec:root} for the case when there are no nonface-leaves.  
I.e., starting with the pair $R(\xi_1)$ and $R(\xi_\ell)$, the ends of the spiral paths are paired up and linked together via connector paths. 
(Recall that, for each $\xi_i$, the endpoints $L(\xi_i)$ and $R(\xi_i)$ are encountered in this order in a clockwise walk along the back rim of the root band, starting, say, at $t$ or at a rim corner.)
The only difference is that, with $t$ in the middle between
$\xi_1$ and $\xi_\ell$, the last pair of spiral path endpoints linked together will be $R(\xi_{\lceil \ell/2 \rceil})$ and $t$ when $\ell$ is odd, and 
$t$ and $L(\xi_{\lceil \ell/2 \rceil + 1})$ when $\ell$ is even. 
The remainder 
of the unfolding is the same. Thus the resulting spiral $\xi$ starts
at $L(\xi_1)$ and ends on the back rim of the nonface-leaf. 
 
\subsubsection{Root Node Unfolding: Two Nonface-Leaves Case}
We now discuss the more complex case when $T$ has two nonface-leaves. In this case, the final spiral path will have its two ends on the back rims of the two nonface-leaves. 
Again let $\xi_1, \xi_2, ..., \xi_\ell$ be the spiral paths corresponding to the face-leaves,
and let $t_1$ and $t_2$ be the spiral paths corresponding to the two nonface-leaves.  We assume that $t_1$ and $t_2$ are labeled such that the number of spiral paths
separating them counterclockwise from $t_2$ to $t_1$ on the root's back rim is at most $\lfloor\ell/2\rfloor$.
(If it is more than $\lfloor\ell/2\rfloor$, then we just switch the labels of $t_1$ and $t_2$.)  
If $\ell > 1$, we further assume that the spiral paths of the face-leaves are labeled so that $t_1$ is in the middle between $\xi_{\lceil \ell/2 \rceil}$ and $\xi_{\lceil \ell/2 \rceil + 1}$. Observe that these labeling rules position $t_2$ on the portion of the rim
counterclockwise between $\xi_1$ and $t_1$.
See~\autoref{fig:twoterminals} for an example with $\ell = 4$, which 
shows the root band, the spiral path endpoints $t_1$ and $t_2$ marked on the back rim of the root band, and the spiral paths  
$\xi_1, \dots \xi_4$ depicted as thick dotted arcs above the root band. 
Note that the counterclockwise ordering  from $\xi_1$ along the rim is: 
$\xi_1, \ldots, t_2, \ldots, t_1$. 

\begin{figure}[htbp]
\centering
\includegraphics[width=0.6\linewidth]{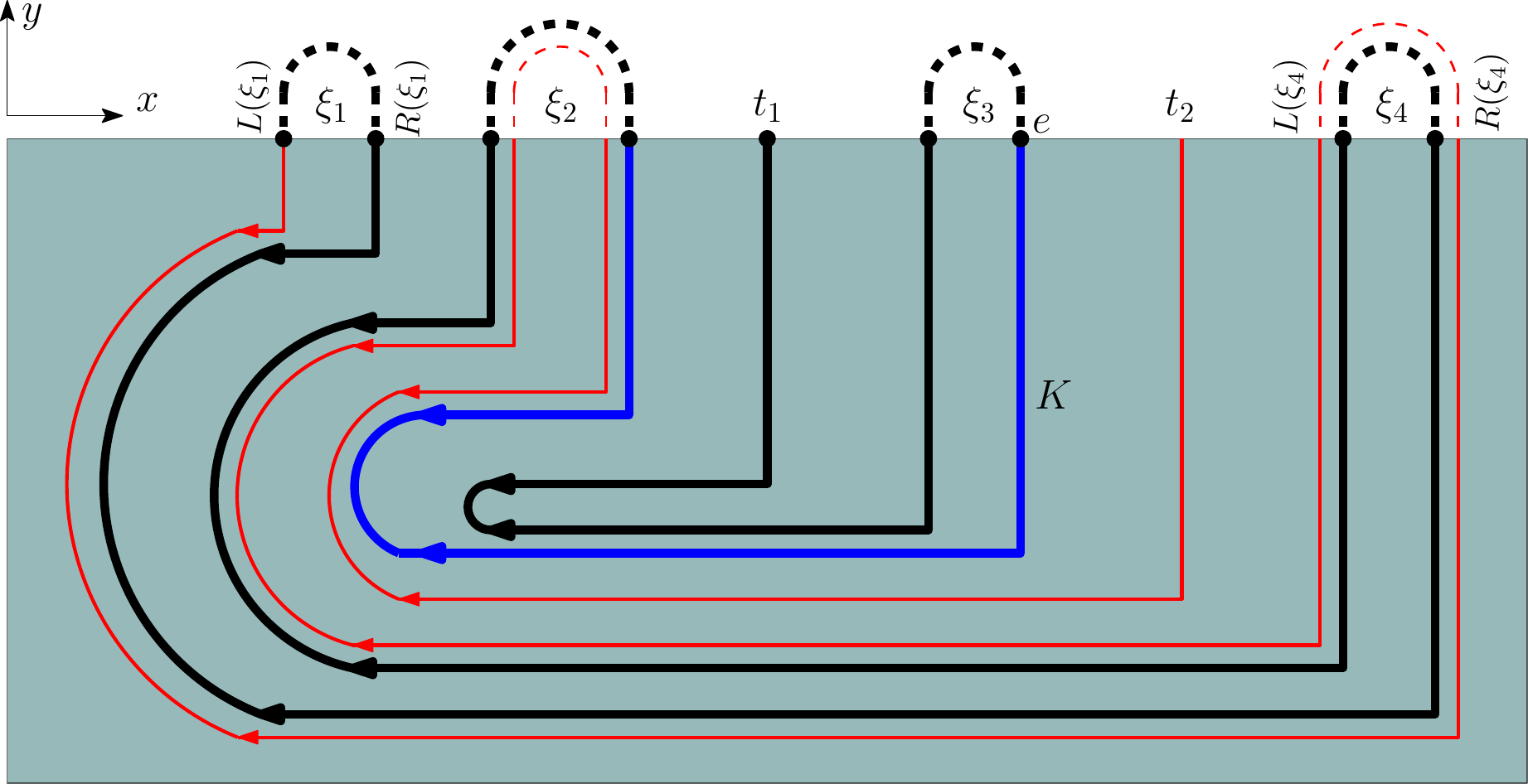}
\caption{Root band, with $\xi$ consisting of the thick (black) and thin (red) paths. The thick (black) path connects $\xi_1, \ldots, \xi_4$ and $t_1$. 
The thin (red) path is the extension of $\xi$ that connects $L(\xi_1)$ to $t_2$. For the two ends of each $\xi_i$,
the dotted paths indicate the spiral path connecting them, which goes down to the associated leaf node and back.} 
\label{fig:twoterminals}
\end{figure}

We begin by linking the spiral paths $\xi_1, \dots, \xi_\ell$ and $t_1$ into a single strip $\xi$ as described in~\autoref{sec:oneterminal} for the case of one nonface-leaf. (See the thick black/blue path marked on the root band in~\autoref{fig:twoterminals}, which shows the portion of
$\xi$ starting at $R(\xi_1)$ and connecting $\xi_1, \dots \xi_4$ and $t_1$.)
Note that $\xi$ has one end at $L(\xi_1)$ and the other end on the back rim of the nonface-leaf corresponding to $t_1$.   
Linking in $t_2$, however, requires some special handling. We will do this by extending $\xi$ from its endpoint $L(\xi_1)$ 
all the way to $t_2$ by following alongside the portion of $\xi$ that tracks 
from $R(\xi_1)$ to $t_1$, until $t_2$ is reached.  (Refer to the thin red path in~\autoref{fig:twoterminals}.) 
At all times, this 
extension is tracing a new path 
on one side of the path it is following. Specifically, starting from
$L(\xi_1)$,  the extension first follows alongside the connector path 
linking  $R(\xi_1)$ to $R(\xi_\ell)$. 
It then follows alongside the spiral path from $R(\xi_\ell)$ to $L(\xi_\ell)$, going all the way down to the leaf node corresponding to $\xi_\ell$ and then back up.  It next follows alongside the connector path linking $L(\xi_\ell)$ to $L(\xi_2)$, and then alongside the spiral path from $L(\xi_2)$ to $R(\xi_2)$, and so on.  This continues until the extension reaches the 
connector path, say $K$,  
that links to the end of the strip piece
located immediately counterclockwise from $t_2$ on the rim.  
This connector is drawn blue in~\autoref{fig:twoterminals}. Because at most 
$\lfloor \ell/2 \rfloor$ spiral paths are located counterclockwise from $t_2$ to $t_1$, the
strip piece
located immediately counterclockwise of $t_2$ 
is either $t_1$ or one of $\{ \xi_{\lceil \ell/2 \rceil + 1},  \xi_{\lceil \ell/2 \rceil + 2}, \dots, \xi_\ell \}$.  If it is one strip piece of the latter, say $\xi_i$, then the connector $K$ links to $R(\xi_i)$. 
Let $e$ be the end of the spiral path ($t_1$ or $\xi_i$) to which 
the connector $K$ links.  

When the extension reaches the connector path $K$, it follows alongside it, 
but instead of following it all the way to $e$, the extension continues 
past $e$ until it reaches $t_2$, at which time it connects with $t_2$ and we are finished. 
This path is drawn 
with a thin red line in~\autoref{fig:twoterminals}: it starts at $L(\xi_1)$,  
follows alongside the thick black path that extends from $R(\xi_1)$ to $e = R(\xi_3)$, and from there 
it reaches $t_2$. The result is the final unfolding spiral $\xi$ with one end on the back face of the nonface-leaf 
corresponding to $t_1$, and the other end on the back face of the nonface-leaf corresponding to $t_2$.  

\begin{figure}[htbp]
\centering
\includegraphics[width=0.98\linewidth]{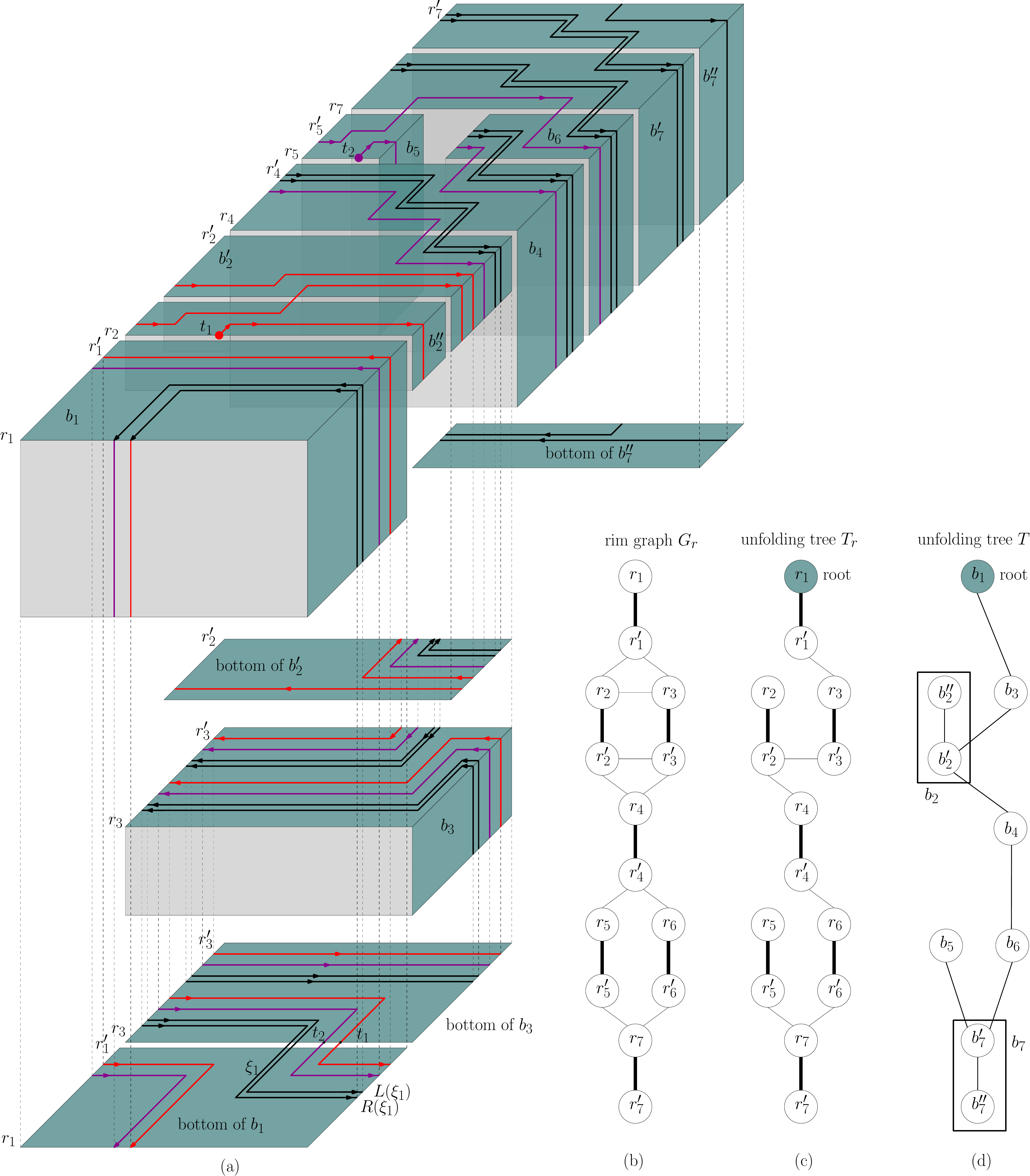}
\caption{Unfolding the genus-2 polyhedron from~\autoref{fig:g2poly}, with slabs slightly separated for clarity (a) spiral path corresponding to the unfolding tree $T$ from (d) of this figure; arrows indicate the direction followed by the unfolding algorithm; (b) rim graph $G_r$; (c) rim unfolding tree $T_r$ (one of several possible) extracted from $G_r$ by the {\sc RimUnfoldingTree} algorithm; (d) unfolding tree $T$ obtained by compressing $b$-edges of $T_r$ and splitting internal nodes with no back children: $b_1$ is the root, $b''_2$ and $b_5$ are nonface-leaves, $b''_7$ is a face-leaf. 
} 
\label{fig:g2tree}
\end{figure}
\autoref{fig:g2tree} shows the complete 
spiral path for the 
genus-$2$ polyhedron from~\autoref{fig:g2poly}, with the slabs slightly separated for clarity. In this example $\ell = 1$, $\xi_1$ corresponds to node $b''_7$ with back face-rim $r'_7$, and the $t_1$ and $t_2$ paths correspond to 
$b''_2$ and $b_5$, with back nonface-rims $r_2$ and $r_5$, respectively. 
(For clarity, the paths $t_1$ and $t_2$ are labeled on the bottom face of $b_3$, as they make the transition to the bottom face of $b_1$.)  

\subsection{Level of Refinement}
Using the Chang-Yen algorithm~\cite{Chang2015}, any genus-0 
orthogonal polyhedron $P$ can be unfolded using linear refinement. Specifically,
they refine each rectangular face of $P$ via 
grid-unfolding 
using a $(2\ell \times 4\ell)$-grid, where
$\ell$ is the number of leaves in $T$, and cuts are allowed along any of these grid
lines.  In the worst case, the algorithm presented here requires at most
twice the level of refinement, i.e., $(4\ell \times 8\ell)$ refinement. 
This
level of refinement is necessary when there are two nonface-leaves in $T$.
In this case, the first part of the spiral path (from $L(\xi_1)$ to the end of $t_1$
located on the back rim of the first nonface-leaf is the same as in the Chang-Yen algorithm,
but with care given to the order in which the spiral paths are connected and with
only one endpoint of $t_1$ at the root.   
(This is the black and blue portions of the path
illustrated in~\autoref{fig:twoterminals}.) Thus no additional
refinement is necessary for
this part of the path.  The doubling of the refinement is due
to the retracing of this path which is needed to extend it from $L(\xi_1)$ to $t_2$,
the spiral path of the second nonface-leaf.
(This is the red portion of the path in~\autoref{fig:twoterminals}.) Because the retracing follows
alongside the existing path, it requires twice the refinement.  Thus, the overall
level of refinement is $(4\ell \times 8\ell)$, which is linear. 

\medskip
\noindent
We conclude with this theorem:
\begin{theorem}
Any orthogonal polyhedron of genus $g \le 2$
may be unfolded to a planar, simple orthogonal polygon,
cutting along a linear grid-refinement.
\label{thm:conclusion}
\end{theorem}
The unfolded planar polygon has $O(n^3)$ edges, for a polyhedron
of $n$ vertices:
$\ell = O(n)$ and there are $O(\ell^2)$ grid edges per face, with $O(n)$ faces.

\section{Conclusion}
It is not evident how to push the techniques common
to~\cite{Damian-Flatland-O'Rourke-2007-epsilon},
\cite{Damian-Demaine-Flatland-2014-delta}, and
\cite{Chang2015} to unfold polyhedra of genus $g \ge 3$,
the next frontier in this line of research.
Both~\autoref{lem:root} (existence of a face-node to serve as root of $T$) and
\autoref{thm:termleaves}
(the {\sc RimUnfoldingTree} algorithm leads to at most $g \le 2$ nonface-leaves)
are crucial in the unfolding algorithm 
described in~\autoref{sec:genus2}.
The final stitching together of the spiral paths
relies on there being at most two nonface-leaves of the unfolding tree $T$.

On the other hand, it is not difficult to unfold the
genus-$3$ polyhedron shown in~\autoref{fig:g0dir}b in an ad-hoc manner.
The challenge is to find a generic algorithm for genus-$3$ and beyond.

\medskip
\noindent
{\bf Acknowledgement.} We thank all the participants of the 31st Bellairs Winter Workshop on Computational Geometry for a fruitful and collaborative environment. In particular, we thank Sebastian Morr for important discussions related to~\autoref{thm:termleaves}, and to the stitching of unfolding strips at the root node. 




\bibliographystyle{alpha}
\bibliography{unfolding}

\end{document}